\documentclass[a4paper,11pt]{article}    
\usepackage{mathptmx,graphics}
\usepackage[authoryear,comma,longnamesfirst,sectionbib]{natbib} 
\usepackage[margin=1in]{geometry}
\usepackage{doi, url, amssymb, amsmath, amsthm, cleveref, enumitem, csquotes, tikz}
\usetikzlibrary{shapes}
\newtheorem{theorem}{Theorem}[section]
\newtheorem{lemma}[theorem]{Lemma}
\newtheorem{corollary}[theorem]{Corollary}
\theoremstyle{definition}
\newtheorem{remark}[theorem]{Remark}

\newcommand{\ex}{\mathbb{E}}
\newcommand{\F}{\bar{F}}
\renewcommand{\L}{\Lambda}
\newcommand{\m}{\mathrm{m}}
\newcommand{\h}{\mathrm{h}}
\newcommand{\e}{\mathrm{\ell}}

\newcommand{\g}{\mathrm{g}}
\renewcommand{\d}{\mathrm{d}}
\newcommand{\du}{\mathop{\d u}}
\renewcommand{\F}{\bar{F}}
\renewcommand{\(}{\left(}
\renewcommand{\)}{\right)}

\providecommand{\keywords}[1]
{\small\textbf{Keywords: } #1}

\providecommand{\class}[1]
{\small\textbf{JEL Classification:} #1}

\begin{document}
\title{On the Equilibrium Uniqueness in Cournot Competition with Demand Uncertainty}
\author{Stefanos Leonardos$^{1}$, Costis Melolidakis$^{2}$  \\
\small $^{1}$Singapore University of Technology and Design, 8 Somapah Rd, 487372 Singapore\\
\small $^{2}$National and Kapodistrian University of Athens, Panepistimioupolis, 15784 Athens, Greece \\
}
\date{}

\maketitle
\begin{abstract}
We revisit the linear Cournot model with uncertain demand that is studied in \cite{Lag06} and provide sufficient conditions for equilibrium uniqueness that complement the existing results. We show that if the distribution of the demand intercept has the decreasing mean residual demand (DMRD) or the increasing generalized failure rate (IGFR) property, then uniqueness of equilibrium is guaranteed. The DMRD condition implies log-concavity of the expected profits per unit of output without additional assumptions on the existence or the shape of the density of the demand intercept and, hence, answers in the affirmative the conjecture of \cite{Lag06} that such conditions may not be necessary.
\end{abstract}
\\
\class{C7}

\section{Introduction}
\cite{Lag06} considers a model of Cournot competition with linear demand -- up to the non-negativity constraint for the market price -- in which the demand intercept is stochastic. Demand uncertainty can make the expected demand sufficiently convex which in turn may lead to a multiplicity of equilibria, \citep{Vi01,La07}. Motivated by this, \cite{Lag06} studies conditions on the distribution of the stochastic demand that guarantee uniqueness of the equilibrium. His main result is that to establish uniqueness, it is sufficient to assume that the distribution has monotone or first decreasing and then increasing (equivalently \emph{bathtub-shaped}, see also \Cref{sub:probabilistic}) hazard rate and that the expected profit per unit of resource is log-concave as the total output approaches zero \citep[Proposition 1]{Lag06}. The latter is achieved by restricting the value at zero of the density of the stochastic demand intercept.\par
In this note, we extend these conditions. We utilize unimodality conditions that we derived in the more general setting of pricing problems with stochastic linear demand \citep{Le18}, and show that equilibrium uniqueness is achieved if the distribution of the demand intercept has the decreasing mean residual demand (DMRD) property or the increasing generalized failure rate (IGFR) property. This is established in \Cref{thm:sufficient}.\footnote{For other related results, see \cite{Leon20,Leo21,Leon21}, preliminary versions of which appear in \cite{Bel18,Kok18}.} The DMRD and IGFR conditions are not comparable, i.e., neither implies the other. Concerning their relationship to the conditions from \cite{Lag06}, the increasing hazard rate property does imply both DMRD and IGFR. However, this is not necessarily true for the decreasing hazard rate nor for the bathtub-shaped property \citep{Gu95}. Hence, the sufficient DMRD and IGFR conditions generalize rather than substitute the existing ones. Given the inclusiveness of the IGFR class of distributions \citep{Pa05,Ba13}, the present conditions cannot be significantly extended. In case that the Cournot oligopolists face zero marginal cost, the uniqueness result can be directly obtained for the class of distributions with the \emph{decreasing generalized mean residual demand} (DGMRD) property and finite $n+1$-th moment, where $n$ is the number of competitors. This is the statement of \Cref{cor:dgmrl}. The DGMRD class includes as a proper subset the IGFR class \citep{Be98,Le18}. Finally, as shown in \Cref{thm:filler}, the DMRD condition implies log-concavity of the expected profit per unit of output without requiring the existence of a density nor any restrictions on its value close to zero (if such a density exists). This answers in the affirmative the conjecture of \cite{Lag06} that the imposed restriction on the value at zero of the density of the stochastic demand may not be necessary. 

\subsection{Outline}
The rest of the paper is structured as follows. In \Cref{sub:model} we restate the model of \cite{Lag06} and in \Cref{sub:probabilistic}, we define all related probabilistic notions and shortly discuss the relationship between the IFR, IGFR, DMRD and DGMRD classes. To make the exposition self-contained, we also provide the relevant results from \cite{Le18}. Our main contribution is presented \Cref{sec:results}. 
\section{Model and Definitions}\label{sec:definitions}
\subsection{Cournot Model}\label{sub:model}
Consider a Cournot market with $n\ge 1$ identical competing firms $i\in\{1,2,\dots,n\}$ that produce a homogeneous good. Each firm has the same constant marginal cost, denoted by $c\ge0$. Firm $i$'s production quantity is denoted by $x_i$ and the total industry output is denoted by $x:=\sum_{i=1}^n x_i$. 
Also, using standard notation, let $x_{-i}:=\(x_1,x_2,\dots,x_{i-1},x_{i+1},\dots,x_n\)$ for any $i=1,2,\dots,n$. The firms face a linear\footnote{Up to the non-negativity constraint.} inverse demand function 
\begin{equation}\label{eq:demand}p\(x\)=\(\alpha-x\)_+=\begin{cases}\alpha-x, &\text{if } x\le \alpha\\0, & \text{otherwise}\end{cases} \end{equation}
where $p\(x\)$ denotes the price for output level $x\ge0$. In \cite{Lag06}, the inverse demand function also includes the slope parameter $b>0$, i.e., it has the form $p\(x\)=\(\alpha-bx\)_+$, with $b>0$. However, the parameter $b$ appears in the analysis -- and in particular, in the first order conditions that determine the Nash equilibrium, cf. equation (1) in \cite{Lag06} -- only as multiplied by $x$. Thus, without loss of generality, we assume throughout this paper that $b=1$ (e.g., by rescaling equation \eqref{eq:demand} or equivalently, by suitably choosing the units of measurement of the output $x$. See also \Cref{rem:b}).\par
The demand intercept or demand level $\alpha$ is a non-negative random variable that takes values in $[0,H]$ for some real number $H>0$ or in $[0,H)$, if $H=+\infty$. We assume that $\alpha$ has an absolutely continuous distribution function $F$, density $f=F'$, and finite expectation $\ex\alpha<+\infty$, that also satisfies $\ex\alpha>c$. We will write $\F\(x\):=1-F\(x\)$ to denote the tail of the distribution of $F$. Firms are risk neutral and maximize their expected payoffs 
\begin{equation}\label{eq:payoff}\pi_i\(x_i ; x_{-i}\):=x_i \(\ex \(\alpha-x\)_+-c\)\end{equation}
with respect to their own output $x_i$, for $i=1,2,\dots,n$. A vector of outputs $\(x_i^*,x_{-i}^*\)$ is a pure Nash equilibrium if $\pi_i\(x_i^*,x^*_{-i}\)\ge \pi_i\(x_i,x^*_{-i}\)$, for all $x_i\ge0$. A pure Nash equilibrium is called symmetric if $x_i^*\equiv x^*$ for some $x^*>0$, for all $i=1,2,\dots,n$. In this case, we will slightly abuse notation and denote the Nash equilibrium with $x^*$. Any symmetric Nash equilibrium satisfies the first order condition 
\begin{equation}\label{eq:condition}\int_{nx^*}^{H}\alpha f\(\alpha\)\mathrm{d}\alpha-c=\(n+1\)x^*\F\(nx^*\),\end{equation}
cf. \cite{Lag06}, Lemma 1. Our scope is to determine conditions on the distribution $F$ of $\alpha$, so that \eqref{eq:condition} has a unique solution, or equivalently that the game has a unique Nash equilibrium in pure strategies. \cite{Lag06} derives such conditions which are reported in \Cref{thm:lagerlof}.

\subsection{Probabilistic Notions}\label{sub:probabilistic}
Let $\h\(x\):=f\(x\)/\F\(x\)$, for $x<H$ denote the \emph{hazard} or \emph{failure rate} of $\alpha$ and $\g\(x\):=x\h\(x\)$, the \emph{generalized failure rate} of $\alpha$ \citep{Be98,La99,Be07}. We will say that $\h$ is \emph{bathtub-shaped} or simply \emph{B-shaped} if there exists an $x_0>0$ (not necessarily unique) such that $\h\(x\)$ is non-increasing for $x<x_0$ and $\h\(x\)$ is non-decreasing for $x>x_0$ \citep{Gu95,Na09}. Also, we will say that $F$ has the \emph{increasing failure rate} (IFR) property, if $\h\(x\)$ is non-decreasing for $x<H$. Similarly, we will say that $F$ has the \emph{increasing generalized failure rate} property or simply that $F$ is IGFR, if $\g\(x\)$ is non-decreasing for $x<H$. Finally, let 
\begin{equation}\label{mrl}\m\(x\):=\begin{cases}\ex\(\alpha-x\mid \alpha>x\)=\displaystyle\frac{1}{\F\(x\)}\int_{x}^{H}\F\(u\)\du, & \text{if } x<H\\0, & \text{otherwise}\end{cases}
\end{equation}
denote the \emph{mean residual demand} (MRD) function of $\alpha$ \citep{Sh07,Lax06}, and $\e\(x\):=\m\(x\)/x$, for $0<x<H$, the \emph{generalized mean residual demand (GMRD)} function \citep{Le18}. Since $\alpha$ is non-negative, we have that $\m\(0\)=\ex \alpha$. We will say that $\alpha$ has the \emph{decreasing mean residual demand} property, or simply that $\alpha$ is DMRD, if $\m\(x\)$ is non-increasing for $x<H$ and similarly, that $\alpha$ has the \emph{decreasing generalized mean residual demand} property, or simply that $\alpha$ is DGMRD, if $\e\(x\)$ is non-increasing for $x<H$. To derive the set of sufficient conditions of the present paper, we will make use of the following properties from \cite{Le18}. 
\begin{lemma}\label{lem:technical}
Let $\alpha$ be a non-negative random variable with $\ex \alpha<+\infty$.
\begin{enumerate}[itemsep=0cm, label=(\roman*).]
\item If $\alpha$ is DGMRD with $\displaystyle\lim_{x\to+\infty}\e\(x\)=\gamma$, and $\beta\ge0$ is a constant, then
 $\gamma<1/\beta$, if and only if $\ex \alpha^{\beta+1}<+\infty$. In particular, $\gamma=0$ if and only if $\ex \alpha^{\beta+1}<+\infty$ for every $\beta>0$.
\item If $\alpha$ is IGFR, then $\alpha$ is DGMRD. The converse is not true in general.
\end{enumerate}
\begin{remark}\label{rem:b}
The normalization of the slope parameter $b$ to $1$ in equation \eqref{eq:demand} does not alter the probabilistic properties of interest of the random variable $\alpha$. This follows from the fact that the IFR, DMRD, IGFR and DGMRD classes of random variables are all closed under multiplication with positive constants (in this case with $1/b$), cf. \cite{Pa05}, Propositions 1 and 4, and \cite{Le18}, Corollary 4.2 and Theorem 4.3.
\end{remark}
\end{lemma}

\subsubsection{The IFR, IGFR, DMRD and DGMRD families}
Based on the above definitions, the IFR property trivially implies the IGFR property and similarly, the DMRD trivially implies the DGMRD property. The IFR property also implies the DMRD property \citep{Ba05}. However, the DMRD and IGFR properties are not comparable \citep{Le18}. Finally, \Cref{lem:technical}-(ii), which has been independently shown by \cite{Be98}, establishes that the DGMRD class is a proper superset of all these classes. These relationships are illustrated in \Cref{fig:venn}. 

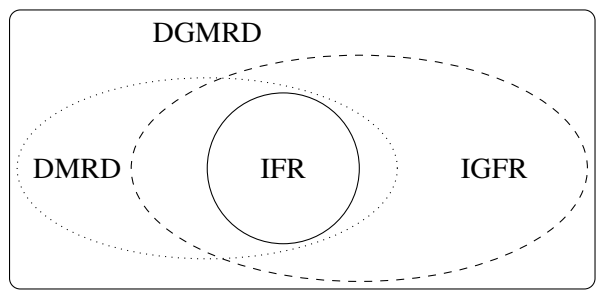
\begin{figure}[!htb]
\center
\begin{tikzpicture}
\draw[rounded corners] (-3.6, -1.6) rectangle (4.1, 2.1) {};
\draw (0,0) ellipse (1cm and 1cm) node {IFR};
\draw[dashed] (1,0) ellipse (3cm and 1.5cm);
\draw (2.2,0) node[right] {IGFR};
\draw[dotted] (-1,0) ellipse (2.5cm and 1.2cm);
\draw (-2,0) node[left] {DMRD};
\draw (-1,1.8) node {DGMRD};
\end{tikzpicture}
\caption{\small Venn diagram illustrating the relationship between the IFR, DMRD, IGFR and DGMRD families. The IFR property (inner circle) implies both the DMRD (dotted ellipse) and the IGFR (dashed ellipse) properties. The DMRD and IGFR classes do not satisfy an inclusion relationship, yet IGFR seems more inclusive in terms of distributions that are interesting for economic applications. DGMRD distributions (outer rectangle) form a proper superset of all these families. Concrete examples of distributions in each family are given in the text.}
\label{fig:venn}
\end{figure}

To get a better understanding of the distributions that are contained in each class, we note that the IFR property is already satisfied by many families of parametric distributions and is, thus, very important for economic applications. As detailed in \cite{Ort19}, the IFR class includes the (lef-truncated) normal distribution, the uniform, the Beta with parameters larger than $1$, the Gamma with shape parameter larger than $1$, the Weibull with exponent larger than $1$, the Log-normal on $\(0, 1\)$, the Logistic, the Laplace, the truncated logistic, the Gumbel (both max and min), the Power function with parameter larger than 1, the chi-squared with degrees of freedom more than 2, the chi with degrees of freedom more than $1$, and the L{\'e}vy model on $\(0, 2/3\)$. \par
The IGFR and DMRD classes properly generalize the IFR class, since they also include distributions with decreasing failure rates (DFR). Such examples are the Beta distribution for values of its parameters less than $1$, the Gamma distribution for any positive value of its shape parameter and the Weibull distribution for any positive values of its exponent \citep{La99}. Further examples of IGFR but not IFR distributions are the Power function on $[0,1]$ with parameter less than $1$, the Lognormal on $\(1,+\infty\)$, the Student's $t$ for certain values of its parameters and the $F$ distribution \citep{Ba05,Ba13}. \par
An important distribution which is IGFR -- and hence DGMRD -- but not DMRD is the heavy-tailed Pareto distribution \citep{La06,Le18}. In the other direction, i.e., to find examples that are DMRD but not IGFR, it suffices to look at distributions that are supported over disjoint intervals. Such distributions can be easily constructed to be DMRD (e.g., proper mixtures of uniform distributions) but which -- inevitably -- are not IGFR, since IGFR distributions need to be defined over continuous intervals \citep{La06}. It is not clear whether DMRD implies IGFR, if we restrict attention to distributions over continuous intervals. Through the same construction, one obtains distributions that are DGMRD but not IGFR. A concrete example of a distribution that is defined over the whole positive half line and which is DGMRD but not IGFR is the Birnbaum-Saunders distribution (which is extensively used in reliability applications) for certain values of its parameters \citep{Le18}. \par
Finally, we note that for pricing and economic applications, the IGFR class seems to exhibit the best trade-off between generality (inclusiveness of distributions) and tractability (handy characterizations), see e.g., \cite{La99,La06} and \cite{Col12} among others. The DGMRD property, although more general, is technically more involved, yet it naturally arises in applications with pricing under demand uncertainty \citep{Leon21}.

\section{Sufficient Conditions}\label{sec:results}
Using the notation and terminology of the previous section, the main results of \cite{Lag06} for the game defined in \Cref{sub:model} can be restated as follows.
\begin{theorem}\normalfont{\citep{Lag06}\textbf{.}}\label{thm:lagerlof}
There exists at least one symmetric Nash equilibrium and no asymmetric equilibria. If $f\(0\)<\left[\ex\alpha-c\right]^{-1}$ and $\h\(x\)$ is (i) monotone or (ii) B-shaped, then this Nash equilibrium is unique. 
\end{theorem}
The statement of \Cref{thm:lagerlof} in \cite{Lag06} does not use the B-shaped terminology, but rather states that the slope of $\h\(x\)$ needs to change sign exactly once starting from negative and turning to positive. It is immediate to see that B-shaped distributions satisfy this property. To proceed with the derivation of new conditions, we will use an equivalent formulation of the first order condition in \eqref{eq:condition}. For this, we define 
\[\L\(x\):=\m\(x\)-c\F\(x\)^{-1}-xn^{-1}, \qquad \text{ for } x\ge 0\] and observe that any symmetric pure Nash equilibrium $x^*>0$ must satisfy $\L\(x^*\)=0$. Since $\L\(0\)=\ex\alpha-c>0$ by assumption and $\lim_{x\to H^-}\L\(x\)=-\infty$, we have that $\L\(x\)$ starts positive and ends negative. This establishes the -- already well known -- existence of symmetric pure Nash equilibria in this model \citep{Am00}. Concerning uniqueness, we have the following sufficient conditions.

\begin{theorem}\label{thm:sufficient}
If $\alpha$ is (i) DMRD or (ii) IGFR, then the symmetric pure Nash equilibrium is unique.
\end{theorem}
\begin{proof} Part (i) is obvious, since in this case $\L\(x\)$ is decreasing for $x>0$. To prove part (ii), it suffices to show that $\L'\(x^*\)<0$ for every equilibrium $x^*$. By continuity of $\L\(x\)$, this implies that $\L\(x\)$ crosses the $x$-axis at most once and hence, it establishes the claim. Taking the derivative of $\L\(x\)$, we obtain that $\L'\(x\)=\h\(x\)\L\(x\)+\frac1n\(\g\(x\)-\(n+1\)\)$. Since $\L\(x^*\)=0$ for every equilibrium, the sign of $\L'\(x^*\)$ is determined by the sign of the term $g\(x\)-\(n+1\)$. To proceed, we consider $\L\(x\)/x$ which is equal to 
\[\L\(x\)/x=\e\(x\)-c\(x\F\(x\)\)^{-1}-n^{-1}\]
Since, by assumption, $\alpha$ is IGFR, \Cref{lem:technical}-(ii) implies that $\e\(x\)$ is decreasing for all $x\ge0$, since $\alpha$ is DGMRD in this case. Moreover, $\(x\F\(x\)\)'=\F\(x\)\(1-\g\(x\)\)$ which implies that $-\(x\F\(x\)\)^{-1}$ is decreasing for all $x\ge0$ such that $\g\(x\)>1$. Hence, as a sum of decreasing functions, $\L\(x\)/x$ is also decreasing for all $x\ge 0$ such that $\g\(x\)>1$. Along with $\L\(x^*\)=0$ for any equlibrium $x^*>0$, this implies that
\[0>\(\L\(x\)/x\)'\Big|_{x=x^*}=\L'\(x^*\)/x^*\]
and hence, that $\L'\(x^*\)<0$ for any $x^*>0$ such that $\g\(x^*\)>1$. If $x^*$ is such that $\g\(x^*\)\le1$, then trivially $\L\(x^*\)=\frac1n\(\g\(x\)-\(n+1\)\)<0$ which concludes the proof of case (ii).
\end{proof}
Neither the DMRD nor the IGFR conditions make use of the the assumption $f\(0\)<\(\ex \alpha-c\)^{-1}$ which is necessary for \Cref{thm:lagerlof} to hold. However, the DMRD, IGFR, monotone decreasing or B-shaped hazard rate conditions are not comparable, since none implies the other \citep{Gu95}, and hence \Cref{thm:sufficient} should be interpreted as complementing rather than substituting \Cref{thm:lagerlof}. In the special case that $c=0$, uniqueness is established under the following condition.
\begin{corollary}\label{cor:dgmrl}
If $c=0$ and $\alpha$ is DGMRD with finite $\(n+1\)$-th moment, then the Nash equilibrium is unique.
\end{corollary}
\begin{proof} In this case $\L\(x\)/x=\e\(x\)-n^{-1}$ and the conclusion follows from \Cref{lem:technical}-(i).
\end{proof}
\begin{remark}\label{rem:conditions}The assumption that $\alpha$ has a density $f$ is not necessary for the DMRD or DGMRD conditions, since they are defined in terms of the cdf $F$. They still apply if $F$ is merely continuous. In this case, we need to rewrite \Cref{eq:condition} in terms of $F$. Additionally, the assumption that $\alpha$ is supported on an interval is also not necessary for DMRD and DGMRD conditions \citep{Le18}. However, for the IGFR to hold, it is necessary that $\alpha$ is supported on an interval \citep{La06}. Since these assumptions increase the complexity without truly generalizing the results in economic terms, we restricted our presentation to distributions with a density. Finally, although DGMRD distributions are a proper superset of IGFR distributions as stated in \Cref{lem:technical}-(ii), \Cref{cor:dgmrl} cannot be viewed as a direct generalization of \Cref{thm:sufficient} due to the moment restriction.
\end{remark}

The condition $f\(0\)<\(\ex \alpha-c\)^{-1}$ which is necessary for \Cref{thm:lagerlof} to hold, guarantees that $\(P\(x\)-c\)$ is log-concave for values of $x$ close to $0$. \cite{Lag06} conjectures that in view of \cite{Am00}, Theorem 2.7, this assumption may not be necessary\footnote{This assumption is rather weak as it is satisfied by any distribution with $f\(0\)=0$ in particular.}. \Cref{thm:filler} shows that this is indeed the case for distributions with the DMRD property which implies the required log-concavity of the expected per unit profit. The DMRD condition does not assume the existence of a density and hence applies to a broader set of distributions. 

\begin{theorem}\label{thm:filler}
If $\alpha\sim F$ is DMRD and $F$ is continuous, then $P\(x\)-c$ is log-concave. 
\end{theorem}
\begin{proof}
$P\(x\)-c$ is log-concave if and only if $\(P\(x\)-c\)'/\(P\(x\)-c\)$ is decreasing \citep{Ba05}. To calculate the derivative of $P\(x\)$, we rewrite $P\(x\)$ as 
\[P\(x\)=\int_{x}^H\(\alpha-x\)f\(\alpha\)d\alpha=\ex\(\alpha-x\)_+=\m\(x\)\F\(x\)\]
Since $F$ is continuous with finite expectation, we have that 
\[\frac{\d}{\d x}\ex\(\alpha-x\)_+=\frac{\d}{\d x}\(\ex\alpha-\int_{0}^{x}\F\(u\)\du\)=-\F\(x\)\]
Hence, 
\[\frac{\(P\(x\)-c\)'}{\(P\(x\)-c\)}=\frac{-b\F\(x\)}{\m\(x\)\F\(x\)-c}=-\cdot\frac{1}{\m\(x\)-\frac{c}{\F\(x\)}}\]
which implies that $\(P\(x\)-c\)'/\(P\(x\)-c\)$ is decreasing if $\m\(x\)-c\F\(x\)^{-1}$ is decreasing. Since $-c\F\(x\)^{-1}$ is decreasing, the claim follows.
\end{proof}
\section*{Acknowledgements}
The authors thank an anonymous reviewer for the quality of their reports and their invaluable feedback, comments and corrections which considerably improved the final version of the paper. Stefanos Leonardos gratefully acknowledges support by a scholarship of the Alexander S. Onassis Public Benefit Foundation.

\bibliographystyle{DeGruyter}
\bibliography{lagerlof_bib}

\end{document}